\documentclass[final,leqno]{siamltex}
\usepackage{amssymb,amsfonts, cite, color,upref,graphicx}
\usepackage{epsfig,amsmath,latexsym}
\usepackage{subfigure}

\newcommand{\eps}{{\varepsilon}}

\newcommand{\rd}{\mathrm{d}}            
\newcommand{\ri}{\mathrm{i}}            

\newcommand{\hy}{{\widehat{y}}}
\newcommand{\hw}{{\widehat{w}}}
\newcommand{\hu}{{\widehat{u}}}
\newcommand{\hv}{{\widehat{v}}}

\newcommand{\op}{{\overline{p}}}

\newcommand{\realpart}[1]{\operatorname{\sf Re}\left(#1\right)}
\newcommand{\impart}[1]{\operatorname{\sf Im}\left(#1\right)}

\newcommand{\eq}[2]{\begin{equation}\begin{split}#1\end{split}\label{#2}\end{equation}}

\newcommand{\eqnn}[1]{\begin{equation}\begin{split}#1\end{split}\nonumber\end{equation}}

\newcommand{\iR}[1]{\int_{\mathbb{R}} #1 \,{\mbox{d}}\xi}

\newtheorem{hypothesis}[theorem]{Hypothesis}

\title{Stability analysis for combustion fronts traveling in hydraulically resistant porous media \thanks{This work was supported by the National Science Foundation through grants DMS-1311313 (A.~Ghazaryan) and DMS-0908074 (S.~Lafortune) and by the Simons Foundation through the Collaboration grant \#233032 (A.~Ghazaryan).}}

\date{\today}

\author{ A. Ghazaryan\thanks{Department of Mathematics, Miami University, Oxford, OH 45056, email:  \texttt{ghazarar@miamioh.edu}, \texttt{mclarnpc@miamioh.edu}}, \and S. Lafortune\thanks{Department of Mathematics, College of Charleston, Charleston, SC 29424,  email: \texttt{lafortunes@cofc.edu}}, \and P. McLarnan\footnotemark[2] }

\begin{document}

\maketitle

\begin{abstract}
We study front solutions of a system that models combustion in highly hydraulically resistant
porous media. The spectral stability of the fronts is tackled by a combination of energy estimates and numerical
Evans function computations.  Our results suggest that there is a parameter regime for which there are no unstable eigenvalues.
We use recent works about partially parabolic systems 
to prove that in the absence of unstable eigenvalues the fronts are convectively  stable.  
\end{abstract}

\begin{keywords} 
traveling waves, combustion fronts, Evans function, spectral stability,  nonlinear stability, partly parabolic systems
\end{keywords}

\pagestyle{myheadings}
\thispagestyle{plain}

\section{Introduction}
This paper is   devoted to the  stability analysis of a traveling front in a combustion model.   {A combustion front is  a coherent  flame structure  that propagates with a constant velocity.  In a physical system it represents an intermediate asymptotic behavior  of a nonstationary combustion wave \cite{B}. To capture combustion waves, mathematical models are often cast as systems of  partial differential equations (PDEs)    posed on infinite domains,  with multiple spatially homogeneous equilibria that correspond to cold states and completely burned states.  Traveling fronts then are  exhibited as a result of a competition of two different states  of the system.   To know the stability stability   of a combustion front  is of importance  for the general understanding of the underlying phenomena, and has implications, for example,  for   chemical  technology and  fire and explosion safety.  From the mathematical point of view,  combustion models are of interest because they are  formulated  as systems of coupled nonlinear PDEs  of high complexity and as such require refined mathematical techniques for their treatment. }

We consider here a system that is obtained from a model for  combustion in highly hydraulically resistant
porous media \cite{Sivash}. The original model \cite{BGSS, Sivash} is considered to adequately capture the rich physical dynamics of the combustion in hydraulically resistant porous media and at the same time be approachable from the mathematical point of view.  
The model consists of a partly parabolic system of two PDEs coupled through a nonlinearity. Partly parabolic systems are systems where some but not all quantities diffuse.
A reduced system can be obtained  from  the  porous media combustion model suggested in \cite{BGSS, Sivash} for a special value of  the ratio of  pressure and molecular diffusivities or  a special value of the Lewis number,  which is  the ratio of  the thermal diffusivity to mass diffusivity \cite{Gordon_review}.

 The existence and uniqueness of combustion front solutions of the system have been studied extensively \cite{Dkh,GKS,GR,GGJ}. The stability of the fronts has not been studied yet  to the authors' knowledge.  We  present here an attempt to provide  the fullest stability analysis possible   that would  allow for  physical interpretation of the results. Our strategy is to perform the stability study for the reduced system. This has the advantage of simplifying the computations. We then show that when the fronts are considered solutions of the full system, no new instabilities are created and our results concerning the reduced system carry over to the full system.

The first  stage of the stability analysis is to  find the spectrum of the operator obtained by linearization of the PDE system about the wave. 
To do so, we first find the essential spectrum through the standard techniques described in \cite{Henry81}. Concerning the point spectrum, 
we find a bound on the eigenvalues (see Theorem \ref{T2}). Our analysis relies on H\"older's and Young's inequalities and is similar to the study performed in  \cite{Varas02}.  
Once the bound is known, we perform a numerical computation of the Evans function \cite{Evans,Jones,Pego92,Yanagida,Alexander90,Sandstede}. The Evans
function is an analytic function of the spectral parameter  defined in some region of the
complex plane. It was used for example in the context of nerve impulse equations \cite{Evans,Jones,Yanagida,Alexander90}, pulse
solutions to the generalized Korteveg-de Vries, Benjamin-Bona-Mahoney and
Boussinesq equations \cite{Pego92},  multi-pulse solutions to
reaction-diffusion equations \cite{Alexander94,Alexander93},
solutions to perturbed nonlinear Schr\"odinger equations
\cite{Kapitula98a,Kapitula98b,Kapitula98c}, near integrable
systems \cite{Kapitula02}, traveling hole
solutions of the one-dimensional complex Ginzburg-Landau equation near
the Nonlinear-Schr\"odinger limit \cite{Kapitula00}, and
 in the context of combustion problems \cite{Gubernov,BGHL,Ghazaryan2013}. 
The Evans function can be expressed in several different ways depending on the size of the stable and unstable manifolds of the linear system.
In this article, we will consider the simplest case where the Evans function is expressed as the scalar product of a solution of the linear system with a solution of the adjoint system \cite{Evans,Jones,Pego92}. 
Numerical computations of the Evans function in this case was first performed by Evans himself \cite{Ev77} and later by
other authors (see for example  \cite{EvansPego,Pego93,SwEl90}).
We perform a numerical computation of the  Evans function and use it to show that there is a parameter regime for which  there is no unstable discrete spectrum that 
 would cause perturbations to the front to grow exponentially fast in time.

Absence of  unstable discrete spectrum  does not necessarily imply stability. This is due to the fact that   (i) the system is partially parabolic, so the linearized operator is not sectorial and  the semigroup it generates is not analytic but  only $C^0$,  and (ii) the front has continuous spectrum that is merely marginally stable.  We overcome these  issues  by using recently obtained results about partially parabolic systems and about fronts with marginal spectrum.

The structure of the paper is as follows. We first introduce the full system, describe existing results related to the combustion fronts, and then  explain how the reduced system is obtained. Then we analytically and numerically study the stability of the combustion fronts. To simplify our computations, we perform our study on the reduced system.  We are then able to show that the results of our analysis of the reduced case carry over to the full system as well.

\section{Background}

In this paper we consider  a system that is related to the combustion model for   highly hydraulically resistant
porous media \cite{Sivash}.  The porous media combustion model  was formulated in \cite{BGSS}  and consists of  a system of coupled PDEs
\eq{ 
T_t-(1-\gamma^{-1})P_t&=\eps T_{xx}+YF(T),  \\
P_t-T_t&=P_{xx}, \\
Y_t&=\eps {\rm Le}^{-1} Y_{xx}-\gamma YF(T), }{phys0}
where $P$, $T$ and $Y$ are  scaled pressure, 
temperature and  concentration of the deficient reactant; 
$\gamma>1$ is the specific heat ratio, $\eps$ is a ratio of
pressure and molecular diffusivities,  ${\rm Le}$ is a Lewis number
and $YF(T)$ is 
the reaction rate. 
The first and the third equations of the system (\ref{phys0}) represent 
the partially linearized conservation equations for energy and deficient 
reactant, while the second one is a linearized continuity 
equation combined with Darcy's law. The original system of physical laws and the deduction of \eqref{phys0} can be found in \cite{GR}.

It is assumed that  function $F(T)$ is of the Arrhenius type with an
ignition cut-off, that is, $F(T)$ vanishes on an interval
$[0,T_{ign}] $ and is positive for $T>T_{ign}$,
\begin{equation}
\label{general}
F(T)=0 \quad \mbox{for} \quad 0\le T < T_{ign} <1, \quad
\end{equation}
and  $F(T)$ is an increasing Lipschitz continuous function, except for a possible discontinuity at the ignition temperature $T=T_{ign}$.

In  \cite{W}, for example,   in  the system \eqref{phys0}  the reaction rate $F(T)$ is  taken to be  a  nonlinear discontinuous function
\begin{equation}
\label{Fdef}
F_{jump}(T) = \left\{\begin{array}{ll}
\mathrm{exp}\left(Z\left\{\frac{T-(1-\gamma^{-1})}{\sigma+(1-\sigma)T} \right\} \right) & T \geq T_{ign}\\
0 & T < T_{ign}.
\end{array} \right.
\end{equation}
Here $Z>0$ is the  Zeldovich number,  and $0<\sigma<1$ represents physical characteristic of the reactant, more precisely it is a ratio of the characteristic temperatures of fresh and burned reactant \cite{Gordon_review}.

 Initial conditions of the following form are usually considered \cite{Gordon_review}, for the process of initiation of combustion,
\begin{equation}T(0,x)=T_0(x), \quad Y(0,x)=1, \quad P(0,x)=0. \label{PDEic}\end{equation}

For realistic materials $\eps$ varies in the range 
$\eps\sim 10^{-2}-10^{-5}$ \cite{Sivash} and therefore it can be treated as a small parameter.
 There is a numerical evidence that
a small thermal diffusivity has no or very minor effect on details
of propagation of the detonation wave \cite{Sivash}.
Hence, \eqref{phys0} is  often   simplified by setting $\eps=0$,
\eq{ 
T_t-(1-\gamma^{-1})P_t&=YF(T),  \\
P_t-T_t&=P_{xx}, \\
Y_t&=-\gamma YF(T). }{phys00}
We note here that this simplification should be done cautiously as (i) this is a valid simplification  only for materials  for with $Le$ is not small, (ii) in general,  introducing the term   $\eps {\rm Le}^{-1} Y_{xx}$ constitutes  a singular perturbation of the system \eqref{phys00} and as such can produce drastic changes in the dynamics.

For both systems,  of interest is  the existence and uniqueness of traveling fronts that asymptotically connect the completely burned state to the state with all of the reactant present. More precisely, 
 one seeks  solutions of
(\ref{phys0}) or (\ref{phys00}), that are solutions  of the 
form $T(x,t)=T(\xi)$, $P(x,t)=P(\xi)$, $Y(x,t)=Y(\xi)$ where $\xi=x-ct$ and $c$
is the a priori unknown front speed.  In other words,  these solutions are  solutions of the traveling wave ODEs 
\begin{eqnarray}
\label{systE}
-cT^{\prime}+c(1-\gamma^{-1})P^{\prime}&=&
\eps T^{\prime\prime}+YF(T), \nonumber \\
P^{\prime\prime}&=&c(T^{\prime}-P^{\prime}), \\
cY^{\prime}+\eps { {\rm Le}^{-1} } Y^{\prime\prime}&=&\gamma YF(T), \nonumber
\end{eqnarray}
or, in the case $\epsilon=0$,
 \begin{eqnarray}
\label{syst0}
-cT^{\prime}+c(1-\gamma^{-1})P^{\prime}&=&
YF(T), \nonumber \\
P^{\prime\prime}&=&c(T^{\prime}-P^{\prime}), \\
cY^{\prime}&=&\gamma YF(T), \nonumber
\end{eqnarray}
 subject to 
 the  boundary conditions
\eq{
P(-\infty)=1,\quad T(-\infty)=1, \quad Y(-\infty)=0, \\
T (+\infty)=0, \quad P(+\infty)=0,\quad Y(+\infty)=1. 
}{bc0}

Traveling wave solutions  of  (\ref{phys00}) are well studied. In \cite{Dkh} the existence and uniqueness of the front is proved using  Leray-Schauder degree theory and estimates on the speed $c$ are obtained. In \cite{GKS},  based on the ODE analysis  it is   shown that 
 a unique value of $c=c_0$  exists for which a front exists that   asymptotically  connects the burned state to the unburned state.   
{Note that in both cases \cite{Dkh,GKS}, the condition $0<T_{ign}<1-\gamma^{-1}$ is necessary for the front to exist.} 
 In \cite{GR} it is shown that the  solutions of the system (\ref{systE}) converge to the solution of (\ref{syst0}) as $\eps\to 0$.  
In \cite{GGJ} it is shown  that the solution of the system (\ref{syst0}) perturbs to a unique solution of  the system (\ref{systE})  at least for small $\eps$. In part, this justifies the simplification of  (\ref{systE}) to (\ref{syst0}), but only if one studies processes  occurring at speeds   that are not slow  (i.e., of order $O(1)$ or faster)  as  this is a local uniqueness result. 


To our knowledge,  for no parameter values the stability of  the traveling fronts for either of these systems has  been studied. In this paper we study the stability of the fronts solutions  of \eqref{phys00}. To do so, we first perform our analysis on a  reduced version of the system, and then show that our results apply to the full system as well.

\section{Reduced systems}  \label{reduction}

It is obvious that the system \eqref{phys00}  has a time conserved quantity 
$T-(1-\gamma^{-1})P+\gamma^{-1}Y  $. We introduce a new variable $$W=  T-(1-\gamma^{-1})P-\gamma^{-1}(1-Y),$$ and transform  the system \eqref{phys00}    to 
\begin{eqnarray} 
\label{phys00_rrr}
W_t&=&0,\nonumber\\
P_t&=&\gamma P_{xx} +\gamma Y F(W +(1-\gamma^{-1})P+\gamma^{-1}(1-Y)),\\
Y_t&=&-\gamma Y F(W+(1-\gamma^{-1})P+\gamma^{-1}(1-Y)).  \nonumber
\end{eqnarray}
With initial conditions \eqref{PDEic},     $W(x,t)= T_0(x)$, but  we instead shall consider this PDE system with such  initial conditions  that  
{\eq{T(x,0)-(1-\gamma^{-1})P(x,0)-\gamma^{-1}(1-Y(x,0))=0.}{incond1}}
 The system  \eqref{phys00_rrr} then  reduces to 
\eq{
P_t&=\gamma P_{xx} +\gamma Y F((1-\gamma^{-1})P+\gamma^{-1}(1-Y)), \\
Y_t&=-\gamma Y F((1-\gamma^{-1})P+\gamma^{-1}(1-Y)). }{phys00_rr}
Recall that this reduction  makes sense  only  if both $\eps$ and  $\eps {\rm Le}^{-1}$ are small. We also can scale $\gamma$ into the time variable and with the abuse of notation ($t\to \gamma t$)  obtain
\eq{
P_t&= P_{xx} +  Y F((1-\gamma^{-1})P+\gamma^{-1}(1-Y)),\\
Y_t&=-  Y F((1-\gamma^{-1})P+\gamma^{-1}(1-Y)).
}{phys00_r}

An alternative approach is to choose a value  of ${\rm Le}$ in order  to reduce the number of variables in \eqref{phys0}.  Indeed,  it is shown in \cite{Gordon_review} that a linear change of variables:
\begin{equation}\label{ls}
\begin{pmatrix}T\\ P\\Y\end{pmatrix} = \begin{pmatrix}h & 1-h& 0\\ (1-\bar \eps)^{-1}&-\bar \eps(1-\bar \eps)^{-1} &0\\0& 0 & 1\end{pmatrix}  \begin{pmatrix}u\\ v\\y\end{pmatrix},
\end{equation}
where $h$ and $\bar \eps$ are defined as $\bar \eps = \eps(1-\mu)^2$, and $h=\mu/(1-\bar \eps)$  where  $\mu$ is a positive solution of $
(1-\mu)(\gamma+\eps\mu)=1$,
$$\mu=(\sqrt{[(\gamma-\eps)^2+4\eps(\gamma-1)]}+\eps-\gamma)/2\eps,$$
 followed by a scaling of the spatial variable $ x\to \sqrt{1-\mu}\,z$,
 leads to 
\begin{eqnarray} \label{rr}
&& u_t= u_{zz}+yF(hu+(1-h) v), \nonumber \\
&& v_t= \bar \eps v_{zz} +(1-v)F(hu+(1-h) v),  \\
&& y_t=\bar \eps((1-\mu){\rm Le})^{-1} y_{zz} -(1-v)F(hu+(1-h) v),\nonumber
\end{eqnarray}
with  $\bar\eps\in(0,1)$ and $ h\in(0,1)$.

Here, it is obvious that if ${\rm Le}=(1-\mu)^{-1}$  then the variable then $y = 1-v$  defines  an invariant set  for this equation. Indeed, if initially  
\eq{y(0,x) = 1-v(0,x),}{incond2}
  then  always $y(t,x) = 1-v(t,x)$. Therefore the system \eqref{rr} reduces to 
\eq{
u_t&= u_{xx}+yF(hu+(1-h)(1-y)),  \\
y_t&= \bar \eps y_{xx} -yF(hu+(1-h) (1-y)). 
}{r}
One also can notice that if $\bar \eps$ is set to be $0$  in \eqref{r}
\eq{
& u_t= u_{xx}+yF(hu+(1-h) (1-y)),  \\
& y_t= -yF(hu+(1-h) (1-y)),  
}{system}
then the linear substitution  \eqref{ls} is simply a linear scaling 
\begin{equation}\label{ls0}
\begin{pmatrix}T\\ P\\Y\end{pmatrix} = \begin{pmatrix}1-\gamma^{-1} & \gamma^{-1}& 0\\ 1&0 &0\\0& 0 & 1\end{pmatrix}  \begin{pmatrix}u\\ v\\y\end{pmatrix},
\end{equation}
since,   when $\bar\eps =0$, $h=\mu$ and, thus,
\eq{h=1-\gamma^{-1}.}{h}
 Also note that    the assumption $v(0,x)=1-y(0,x)$ is equivalent to assuming that
$$T(0,x) =h u(0,t)+(1-h)v(0,t)=hP(0,t)+(1-h)(1-Y(0,x))$$ or in the notations of the previously described reduction it is equivalent to assuming that $W(x,0)=0$.
 Therefore,  system  \eqref{system} is equivalent to  \eqref{phys00_r} with $W(x,0)=0$. 

To summarize, there are two different ways to reduce  \eqref{phys0}. One is based on setting $\eps$ and $\eps {\rm Le}^{-1}$ equal to $0$, and restricting the set of initial conditions to \eqref{incond1}, thus obtaining \eqref{phys00_rr}. The other  is based on choosing a specific ${\rm Le}^*={\rm Le}(\eps,\mu,\gamma)$ as described above and restrict \eqref{phys0} to the set of initial conditions \eqref{incond2}, thus obtaining \eqref{r}.  We point out that  setting $\bar\eps=0$ in \eqref{r} brings \eqref{r} to the same system  \eqref{phys00_rr} (now \eqref{system}) since when $\epsilon=0$ the choice of ${\rm Le}$ is not important for any fixed value of it.  In addition, under the assumption of $\eps=0$ (respectively, $\bar\eps=0$) \eqref{incond1} and \eqref{incond2} describe the same set of initial conditions.  The nonlinearity $F$ can be then defined as in  \eqref{general} or \eqref{Fdef} with $F(T)=F(hu+(1-h) (1-y))$, where $h$ is simply $1-\gamma^{-1}$.

In this article, we first study the stability of the fronts as solutions of the reduced system \eqref{system} which holds when the initial condition \eqref{incond2} is imposed\footnote{The authors thank P. Gordon for bringing their attention to the importance of the problem of  front stability in the reduced system.}. However, as we show in Section \ref{Conc}, when the condition \eqref{incond2} is relaxed, no additional instabilities are created and our results obtained for the reduced system still hold. Therefore, our stability analysis applies to the front solutions of the system \eqref{phys00} without additional assumptions.


\section{{{Explicit form for the discontinuous term}}}

We consider the system \eqref{system}
with $h  \in (0,1)$.
 Part of our stability analysis is based on numerical calculation of the spectrum of the operator obtained by linearizing \eqref{system}. Therefore we cannot work with a general form of the nonlinearity \eqref{general} but have to be more specific. 
 We aim at working with the discontinuous $F$ given in \eqref{Fdef}, which, in the case of the reduced system \eqref{system} and with the relation \eqref{h}, becomes
 \begin{equation}
\label{Fdefr}
F_{jump}(T) = \left\{\begin{array}{ll}
\mathrm{exp}\left(Z\left\{\frac{T-h}{\sigma+(1-\sigma)T} \right\} \right) & T \geq T_{ign}\\
0 & T < T_{ign}.
\end{array} \right.
\end{equation}
However, for the  stability analysis of the {{front}}, it will be necessary to remove the discontinuity and consider a smooth $F$, which is defined like $F_{jump}$ everywhere except for a small interval $(T_{ign} , T_{ign}+2\delta)$ where the function is modified so as to go to zero in a smooth and monotonic fashion. One way this can be done is to consider
 \begin{equation}
\label{Fdefrc}
F_\delta(T) = \left\{\begin{array}{ll}
\mathrm{exp}\left(Z\left\{\frac{T-h}{\sigma+(1-\sigma)T} \right\} \right) & T \geq T_{ign}+2\delta\\
\\
F_{jump}(T_{ign}+2\delta)\,H_\delta(T-T_{ign}-\delta)&  T_{ign}\leq T < T_{ign}+2\delta\\
\\
0 & T < T_{ign},
\end{array} \right.
\end{equation}
where $H_\delta$ is a smooth approximation  of the Heaviside function, which has the property that
$$
\lim_{\delta\rightarrow 0^+}H_\delta =H,\;H_\delta(x)=1,\;\;{\mbox{for}}\;\; x>\delta, \;\;H_\delta(x)=0,\;\;{\mbox{for}}\;\;x<\delta.$$
The limit above is understood in the distributional sense. A simple example of how this $H_\delta$ can be chosen is given by 
\eq{
H_\delta(x)=\frac{1}{1+{e^{{\frac {4x\delta}{ x^2-\delta^2
 }}}}},\;\;{\mbox{for}}\;\;|x|<\delta.
}{Hex}

While the  stability analysis necessitates the use of a smooth version of $F$, the numerical computation are done for small values of $\delta$ and the nonlinear stability analysis performed in Section \ref{nonl} are valid for any $\delta$. Furthermore, we discuss below the fact that the front solution of   \eqref{system} with $F=F_\delta$ approaches the front solution when $F$ is chosen to be $F_{jump}$ as $\delta\rightarrow 0^+$. Additionally, from the small size of what the jump of $F_{jump}$ would be   in the specific example discussed in Section \ref{SpSt}, it is clear that smoothing the function $F$ have very little effect on the numerics. 
We can thus interpret the stability results of this paper to be valid t when  $\delta$ in \eqref{Fdefrc} is very small.

We now want to prove that the front solution of  \eqref{system} with $F=F_\delta$ approaches the front solution of the discontinuous system as $\delta\rightarrow 0^+$.
The construction of the front solution in \cite{GKS} is based on the existence and uniqueness of a smooth  solution of a boundary value problem in which the dependent variable is $y$ and the independent is $u$. Using the notation corresponding to the system \eqref{system}, the boundary problem is obtained from the traveling wave reduction given below in \eqref{TWZ} and is given by
$$\frac {d y}{du} = \frac{\gamma yF(hu+(1-h)(1-y))}{c^2(1-u-y)},$$
with boundary conditions
$y\left(\frac{T_{ign}}{h}\right)=1$, $y(1)=0$.
Between the boundary points, $u=\frac{T_{ign}}{h}$ and  $u=1$,  the argument of $F$ (given by $hu+(1-h)(1-y)$) varies monotonically
from $T_{gin}$ to $1$. Given that $F$ is Lipschitz continuous on the interval $[T_{ign}, 1]$ and given that $F$ in addition depends on a parameter $\delta$ (such as suggested in \eqref{Fdefrc} with $H_\delta$ given in \eqref{Hex}),   smoothly, and in a way that the conditions in the theorems of \cite{GKS}  are satisfied ({Lipschitz continuity for all $T\geq T_{ign}$,  $F(T)=0$ for $T<T_{ign}$, and $0<T_{ign}<h$}), then the front   and the value of $c$ for which it exists    depend on  $\delta$ in a continuous way.


\section{Front Solution}


We write the system (\ref{system}) in variable $\xi=x-ct$ and $t$:
\begin{equation}\label{e:6}
\begin{aligned}u_t &= u_{\xi\xi} + c u_\xi + yF(w), \\
y_t &= cy_\xi-yF(w),
\end{aligned}
\end{equation}
where, for brevity,
$w\equiv hu+(1-h)(1-y).$
The front solution then becomes a stationary solution of \eqref{e:6}, i.e. a solution of the dynamical system
\begin{equation}
\label{TWZ}
\begin{aligned}
&u_{\xi\xi} + cu_\xi + yF(w) = 0,\\
& cy_\xi - yF(w) = 0.
\end{aligned}
\end{equation}
We denote  by $(u,y)=(\widehat{u},\widehat{y})$ the front solution, which satisfies
\begin{equation}\label{bc} (\hu,\hy) \to (1,0) \text{ as } \xi\to  -\infty, \text{ and } (\hu,\hy) \to (0,1)  \text{ as } \xi\to  \infty.\end{equation}

In the next sections we study the stability of the front  $(\widehat{u},\widehat{y})$. The most important part of the stability analysis is  to find   the spectrum of the linear operator produced by the front.

\section{Linearization}


In this section, we  analytically study the spectrum of the linear operator arising from the linearisation
of \eqref{e:6} about the front solution. We will first find the essential spectrum and then find a bound on any eigenvalue with positive real part.
We point out that neither the procedure of finding the essential spectrum nor the numerical calculation of the  discrete spectrum are sensitive  with respect the choice of $\delta$ in \eqref{Fdefrc}. Indeed, the computation of the continuous spectrum will be the same for any choice of $\delta$ in the smooth version of $F$. Furthermore, as discussed before, the choice of $\delta$ is expected to have little effect on the numerical computations if $\delta$ is chosen to be small enough.  However, in order to have a linear operator that is smooth at every point along the front, we  consider the smooth version.


The eigenvalue problem arising from the linearization of \eqref{e:6} about $(\widehat{u},\widehat{y})$ reads
\begin{equation}
\label{EigProb}
\begin{aligned}\lambda p &= p_{\xi\xi} + c p_\xi + F_w(\hw)\,\hy\,(hp-(1-h)q)+F(\hw)q, \\
\lambda q &= c q_\xi-F_w(\hw)\,\hy\, (hp-(1-h)q)-F(\hw)q.
\end{aligned}
\end{equation}

Alternatively, this system  can be written  as 
${\mathcal{L}}W=\lambda W$, $W=(p,q)^T$,
where
\begin{equation}
{\mathcal L}= \begin{pmatrix}1& 0 \\
0 &0
\end{pmatrix}\frac{\rd^2}{\rd \xi^2}+c\frac{\rd}{\rd \xi}+ 
\begin{pmatrix}
F_w(\hw)\,\hy\,h& F_w(\hw)\,\hy\,(1-h)+F(\hw) \\
-F_w(\hw)\,\hy\,h & -F_w(\hw)\,\hy\,(1-h)-F(\hw)
\end{pmatrix}.
\label{oPDEf}
\end{equation}

Generally speaking, a traveling wave  is called spectrally stable in $L^2(\mathbb{R})$ if the spectrum of the linear operator $\mathcal L$ on $L^2(\mathbb{R})$  is contained in the half-plane $\{\mathrm{Re} \, \lambda \le -\nu\}$ for some $\nu>0$ with the  exception for a simple zero eigenvalue which is generated by the translational invariance. 

If the discrete spectrum of the linear operator $\mathcal L$ in $L^2(\mathbb{R})$ is stable, i.e it is contained in the half-plane  $\{\mathrm{Re} \lambda \le -\nu\}$ for some $\nu>0$, except for a simple eigenvalue at 0, but the essential spectrum  has nonempty intersection with the imaginary axis then the wave sometimes is called spectrally unstable due to essential spectrum.

The system \eqref{e:6} is an example of a partly parabolic system: the $v$-component does not diffuse.   Partly parabolic systems  or, as they are also called, partly dissipative \cite{SY}, or   partially degenerate \cite{V} are not as well studied as parabolic problems. The stability analysis in parabolic problems is based on the fact that the linearized operators are sectorial and thus generate analytic semigroups \cite{Henry81}. In the  case of the system   \eqref{e:6},   
the linear operator $\mathcal  L$  is not sectorial, so it  generates a $C^0$ semigroup. Hence, unlike the case of sectorial operators,  generally speaking, the spectral stability of the wave does not simply imply    either the linear stability of the traveling wave, i.e., the
decay of the solutions of the linearized system, or   the orbital stability of the wave, as it does in case  of sectorial operators  \cite{Henry81}.
Recall that a wave is called orbitally stable   if a solution  that starts near the wave   stays  close to the family of translates of the wave, in some norm  not necessarily the same as the norm of the original space.
If, in addition, it converges (in some norm) to a particular  translate of the wave then the wave is called orbitally or nonlinearly  stable with asymptotic phase.

In any case, the stability analysis of a front starts  with finding  its spectrum.  The spectrum of a front consists of  the isolated  eigenvalues of finite multiplicity  of $\mathcal L$ and the continuous part of the spectrum that is called essential.

\subsection{Essential Spectrum}
\label{Ess}
We define
${\mathcal{L}}^{\pm \infty}=\lim_{x\rightarrow \pm\infty}{\mathcal{L}}$ 
and  compute them by inserting $\hy=1$ and $\hw=0$ into (\ref{oPDEf}) for ${\mathcal{L}}^{+\infty}$ and 
$\hy=0$ and $\hw=1$ for ${\mathcal{L}}^{-\infty}$:
$${\mathcal{L}}^{+\infty}=
 \begin{pmatrix}\frac{\rd^2}{\rd \xi^2}+c\frac{\rd}{\rd \xi}& 0 \\
0 & c\frac{\rd}{\rd \xi}
\end{pmatrix},\;\;
{\mathcal{L}}^{-\infty}=
 \begin{pmatrix}\frac{\rd^2}{\rd \xi^2}+c\frac{\rd}{\rd \xi}& e^{(1-h)Z} \\
0 & c\frac{\rd}{\rd \xi}-e^{(1-h)Z}
\end{pmatrix}.
$$
 The boundaries of the essential spectrum of ${\mathcal{L}}$ is given 
by the spectra of constant coefficient operators ${\mathcal{L}}^{\pm\infty}$ (see \cite{Henry81}, Theorem A.2, p.\ 140). 
We compute these spectra using 
Fourier analysis, which amounts to computing the value of $\lambda$ for which
\begin{equation}
\label{det}
{\mbox{det}}\left({\mathcal{L}}^{\pm\infty}_\sigma-\lambda I\right)=0,
\end{equation}
where ${\mathcal L}^{\pm\infty}_\sigma$ are obtained from ${\mathcal L}^{\pm\infty}$ by replacing $\frac{\rd}{\rd \xi}$ with $i\, \sigma$.
The set of curves defined by the equation \eqref{det}
\eq{
\{\lambda&=-\sigma^2+ci\sigma; \,\,\sigma\in \mathbb R\} \cup\{\lambda = ci\sigma;\,\,\sigma\in\mathbb R\} , \\
\{ \lambda&=-\sigma^2+ci\sigma; \,\,\sigma\in \mathbb R\} \cup\{\lambda =c i\sigma - e^{(1-h)Z)};\,\,\sigma\in\mathbb R\}
}{es1}
 divides the complex plane into regions that are either covered by spectrum or, otherwise, contain only discrete eigenvalues. Here the first set is due to behavior at $+\infty$ and the second due to behavior at $-\infty$.
It is found that the rightmost curve defined by (\ref{det}) is $\lambda=\ri \sigma$, i.e. the imaginary axis. Therefore the essential spectrum of $\mathcal L$ contains the imaginary axis as its rightmost boundary. The open left side of the 
complex plane can only contain point spectrum, i.e. isolated eigenvalues with finite multiplicity \cite{Simon06}.

So on the spectral level the front is unstable at least due to the essential spectrum on the imaginary axis. We shall numerically show in the following sections that for some parameter values,  the essential spectrum covering all of the imaginary axes   is the only unstable spectrum that the front has. We test the essential spectrum using an exponential weight $e^{\alpha \xi}$ to see if at least on the spectral level the instability is convective. Indeed, introducing exponential weight in the eigenvalue problem by using new variable $(\widetilde p, \widetilde q)= e^{\alpha \xi}(p,q)$ yields
\begin{equation}
\label{EigProbweight}
\begin{aligned}\lambda \widetilde p &= \widetilde p_{\xi\xi} + (c-2\alpha)  \widetilde p_\xi + F_w(\hw)\,\hy\,(h\widetilde p-(1-h)\widetilde q)+F(\hw)\widetilde q +(- c\alpha +\alpha^2)\widetilde p, \\
\lambda \widetilde q &= (c-2\alpha) \widetilde q_\xi-F_w(\hw)\,\hy\, (h\widetilde p-(1-h)\widetilde q)-F(\hw)\widetilde q -c\alpha \widetilde q.
\end{aligned}
\end{equation}
We use the constant coefficient operators $ {\mathcal{L_{\alpha}}}^{\pm\infty}$ obtained from ${\mathcal L}^{\pm\infty}$ by replacing $\frac{\rd}{\rd \xi}$ with $i\, \sigma -\alpha$ and find the following curves for the boundaries of the essential spectrum:
\eq{
\{\lambda&=(i\sigma-\alpha)^2+c(i\sigma-\alpha); \,\,\sigma\in \mathbb R\} \cup\{\lambda = c(i\sigma-\alpha);\,\,\sigma \in\mathbb R\} , \\
 \{\lambda&=(i\sigma-\alpha)^2+c(i\sigma-\alpha); \,\,\sigma\in \mathbb R\} \cup\{\lambda =c (i\sigma-\alpha) - e^{(1-h)Z)};\,\,\sigma\in\mathbb R\}.
}{es2}
{ The right most boundary   $\lambda= -\sigma^2 + \alpha^2 -\sigma\alpha + (c-2\alpha)i\sigma$   of this set, as a curve on the complex plane,  is depicted on Fig.~\ref{fig:esssp}.  An  exponential weight $e^{\alpha\xi}$ with a rate  such that $\alpha^2-c\alpha<0$ but  $c-2\alpha>0$  pushes these curves to the left of the imaginary axis. }   We formulate this fact as a lemma.

\begin{figure}
\begin{center}
\scalebox{.21}{{\includegraphics{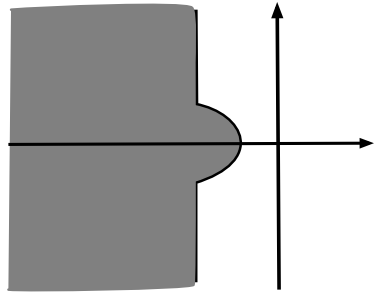}}}
\caption{\label{fig:esssp}  Essential spectrum in the  space with an exponential weight. }\end{center}
\end{figure}

\begin{lemma}
\label{lem:essp}
Consider the system \eqref{system}  with $c=c(\gamma, h, \sigma )$ such that  a traveling front exists that  the boundary conditions \eqref{bc} are satisfied.   Suppose $0<\alpha<c/2$.  Then there exists $\nu>0$ such that the essential spectrum of $\mathcal L$ on $E_\alpha$ is contained in a half-plane $\mathrm{Re} \,\lambda \le -\nu<0$.
\end{lemma}

\subsection{Bound on the point spectrum}

We assume that $\lambda$ is an eigenvalue of ${\mathcal{L}}$ defined in  (\ref{oPDEf}) with $\realpart{\lambda}>0$. We want to find a bound on the norm of $\lambda$ in the complex plane. The argument we use here follows closely the argument presented in Section 3 of \cite{Varas02}. 

We first want to prove the following lemma, which provides a bound on $\lambda$.
\begin{lemma}\label{Lemma1}
Assume $\lambda$ is an eigenvalue with positive real part for the problem (\ref{EigProb}) on $L^2(\mathbb{R})$. Then
\eq{|q(\xi)|\leq h_2(\xi)\|p\|_{L^2},}{firstin}
where 
\eq{h_1(\xi)&=\frac{1}{c}\int_{\xi}^\infty{( (1-h)F_w(\hw(z))\hy(z)-F(\hw(z))){\mbox{d}}z}
}{h1h2}
and
\eq{h_2(\xi)= \frac{h}{c}\sqrt{\int^{+\infty}_{\xi}( F_w(\widehat w(z))\widehat y(z) )^2e^ {-2(h_1(\xi)-h_1(z))}dz}.}{h2}
\end{lemma}
\begin{proof}
We first observe that the expression defining $F$  in (\ref{Fdefrc})  for $w\geq T_{ign}$ is a nondecreasing function and
thus $F_w$ is nonnegative for all $w$, $\sigma$, and $h$ in $(0,1)$.
We now rewrite the second equation of (\ref{EigProb}) as a first order non-homogeneous linear equation for $q$:
\eq{c q_\xi+\left((1-h) F_w(\hw)\,\hy-F(\hw)-\lambda\right)q=h\,F_w(\hw)\,\hy\, p,}{qeq}
which can be rewritten as
\eq{\left(q(\xi)\exp\left(-h_1(\xi)-\frac{\lambda}{c}\xi\right)\right)'=\frac{h}{c}F_w(\hw(\xi))\hy(\xi)\,   \exp\left(-h_1(\xi)-\frac{\lambda}{c}\xi\right)p(\xi),}{qeqint}
where $h_1$ is as in the lemma. Integrating both sides of (\ref{qeqint}) and solving for $q$ gives
\eq{q(\xi)=-\frac{h}{c}\exp\left(h_1(\xi)\right)\int_\xi^\infty  F_w(\hw(z)) \hy(z)\,    \exp\left(-h_1(z)+\frac{\lambda}{c} (\xi-z)\right)p(z)\,{\mbox{d}}z,}{qsolved}
which implies
\begin{eqnarray*}|q(\xi)|& \leq& \frac{h}{c} \exp(h_1(\xi))\int^{+\infty}_{\xi} F_w(\widehat w(z))\widehat y(z) \exp (-h_1(z)+\frac{\lambda}{c}(\xi-z))p(z)dz\\
&\leq &  \frac{h}{c} \exp(h_1(\xi))\int^{+\infty}_{\xi} F_w(\widehat w(z))\widehat y(z) \exp (-h_1(z))p(z)dz \\
&\leq &  \frac{h}{c}  \|p\|_{L^2}  \sqrt{\int^{+\infty}_{\xi}( F_w(\widehat w(z))\widehat y(z) )^2e^ {-2(h_1(\xi)-h_1(z))}dz},  \end{eqnarray*} 
where we have used H\"older's inequality to find the last expression. The last inequality gives (\ref{firstin}).
 \qquad\end{proof}

\begin{theorem}
\label{T2}
If  $\lambda$ is an eigenvalue with positive real part, then we have the following inequalities:
{{\eq{\realpart{\lambda}\leq h_3+h_4\;\;{\mbox{and}}\;\;\left|{\lambda}\right| \leq  \frac{c^2}{4}+\sqrt{2}\,h_3+h_4,}{rein126z}}}
where 
{{$h_3= \sqrt{{{}}\iR{\left(F(\hw)+(1-h){F_w(\hw)\,\hy }\right)^2h_2(\xi)^2}}$ and $h_4=h\sqrt{\iR{\left(F_w(\hw)\,\hy\right)^2}}.$}}

\end{theorem}
\begin{proof}
We multiply the first equation of (\ref{EigProb}) by $\overline{p}$ and we then integrate and get
\eq{
\lambda \iR{|p|^2}=c\iR{p_\xi \overline{p}}-\iR{|p_\xi|^2}&+\iR{\left(F(\hw)-(1-h)F_w(\hw)\,\hy\right) \,q\,\op}\\&+h\iR{F_w(\hw)\,\hy\,|p|^2}.
}{fez}
We take the real part of (\ref{fez}) to get
\eq{\realpart{\lambda}\iR{|p|^2}=-\iR{|p_\xi|^2}+\iR{\left(F(\hw)-(1-h){F_w(\hw)\,\hy }\right)\,\realpart{q\,\op}}\\+h\iR{F_w(\hw)\,\hy\,|p|^2}.}{rez}
From (\ref{rez}), we obtain the inequality
$$\realpart{\lambda}\iR{|p|^2}\leq \iR{\left(F(\hw)+(1-h){F_w(\hw)\,\hy }\right)\,|\realpart{q\,\op}|}+h\iR{F_w(\hw)\,\hy\,|p|^2}.$$
Using the fact that $|\realpart{q\,\op}|\leq |p\,q|$ and Lemma \ref{Lemma1}, the inequality above implies
$$\realpart{\lambda}\|p\|_{L^2}^2\leq \|p\|_{L^2}\iR{\left(F(\hw)+(1-h){F_w(\hw)\,\hy }\right)h_2(\xi)\,|p|}+h\iR{F_w(\hw)\,\hy\,|p|^2}.$$
We then prove the first part of (\ref{rein126z}) by using H\"older's inequality. 


{{
We now take the imaginary part of (\ref{fez}) to obtain
\eq{
\impart{\lambda} \iR{|p|^2}=-\ri c\iR{p_\xi \overline{p}}+\iR{\left(F(\hw)-(1-h)F_w(\hw)\,\hy\right) \impart{\left(\,q\,\op\right)}}.
}{fez2} 
We then add (\ref{rez}) and the modulus of (\ref{fez2}) to obtain
\eqnn{\left(\realpart{\lambda}+\left|\impart{\lambda}\right|\right)\iR{|p|^2}&\leq
c\iR{|p_\xi| |\overline{p}|}-\iR{|p_\xi|^2}+h\iR{F_w(\hw)\,\hy\,|p|^2}\\&+\iR{\left(F(\hw)+(1-h){F_w(\hw)\,\hy }\right)\,\left(\left|\realpart{q\,\op}\right|+\left|\impart{q\,\op}\right|\right)}.}
Next, we  apply Young's inequality to get ${c\,|p_\xi ||p|\leq \frac{c^2|p|^2}{4}+|p_\xi|^2}$. Using $|\lambda|\leq \realpart{\lambda}+\left|\impart{\lambda}\right|$, $\left|\realpart{q\,\op}\right|+\left|\impart{q\,\op}\right|\leq \sqrt{2}|p||q|$,  Lemma \ref{Lemma1} and H\"older's inequality, we prove the second part of (\ref{rein126z}).

}}

\qquad\end{proof}

\section{Numerical Computations}

Note that for all the numerical computations presented in this section, the function $F$ is chosen to be in the form \eqref{Fdefrc}. We use the approximation of the Heaviside function to be the one given in \eqref{Hex} and we choose $\delta$ to be 0.01.

\subsection{Front Solution}

In order to compute the front solution of (\ref{TWZ}) numerically, we perform the following manipulations. 
We add the two equations of (\ref{TWZ}) together  and integrate. We obtain the following system of equations
\begin{equation}
\begin{aligned}&u_\xi + cu + cz = r \\
& y_\xi -\frac{1}{c} yF(hu+(1-h)(1-y)) = 0,
\end{aligned}
\end{equation}
where $r$ is constant. Since front solution $(\hu,\hv)$
converges to (0,1) as $\xi\rightarrow\infty$, we set $r=c$.  We are thus left with the following two-dimensional 
dynamical system to integrate:
\begin{equation}
\begin{aligned}u_\xi &= c(1-y-u) \label{front_eps0} \\
y_\xi &=  \frac{1}{c}yF\left(hu+(1-h)(1-y)\right).
\end{aligned}
\end{equation}
We are interested in the heteroclinic orbit that connects the the fixed points $(1,0)$ and $(0,1)$. The unstable direction of the point (1,0) is given by 
$(c^2,-c^2-e^{(1-h)Z})$ with corresponding eigenvalue $e^{(1-h)Z}/c$. To find the front solution for given values of $h$, $\sigma$, {$T_{ign}$, $\delta$}, and $Z$, 
we use a simple shooting method and make incremental changes in the value of $c$ until  an appropriate connection to the fixed point (0,1) is found. {{Figure \ref{SpeedFigure} shows how the speed $c$ varies as a function of the various parameters.
In the case where $h=0.3$, $Z=6$, $\sigma=0.25$, $T_{ign}=0.003$, {and $\delta=0.01$}, it is found that $c=1.4684$. 
}} Figure \ref{FigFront} shows the front solution as a function of $\xi$ in that case. 
Note that if $F$ would be chosen to be $F_{jump}$, the size of the jump  at the discontinuity would be $F(T_{ign})=8.55\times 10^{-4}$ in this case. Therefore, there would be very little change in the numerics if the computations were done using the discontinuous version of $F$ given in \eqref{Fdefr}.

\begin{figure}
\subfigure[]{
{\includegraphics[width=180pt]{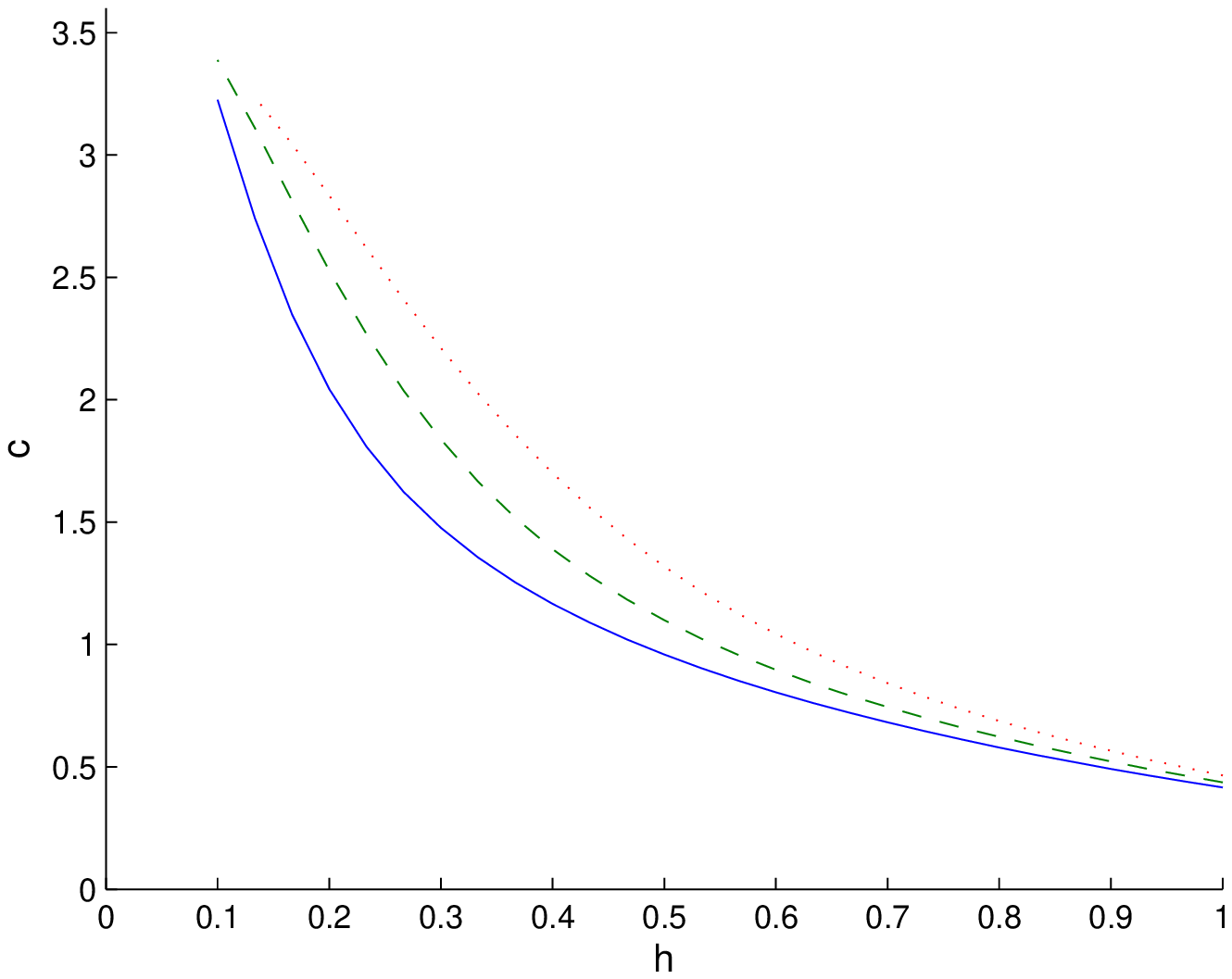}}  
\label{fig:subfig1}
}
\subfigure[]{
{\includegraphics[width=180pt]{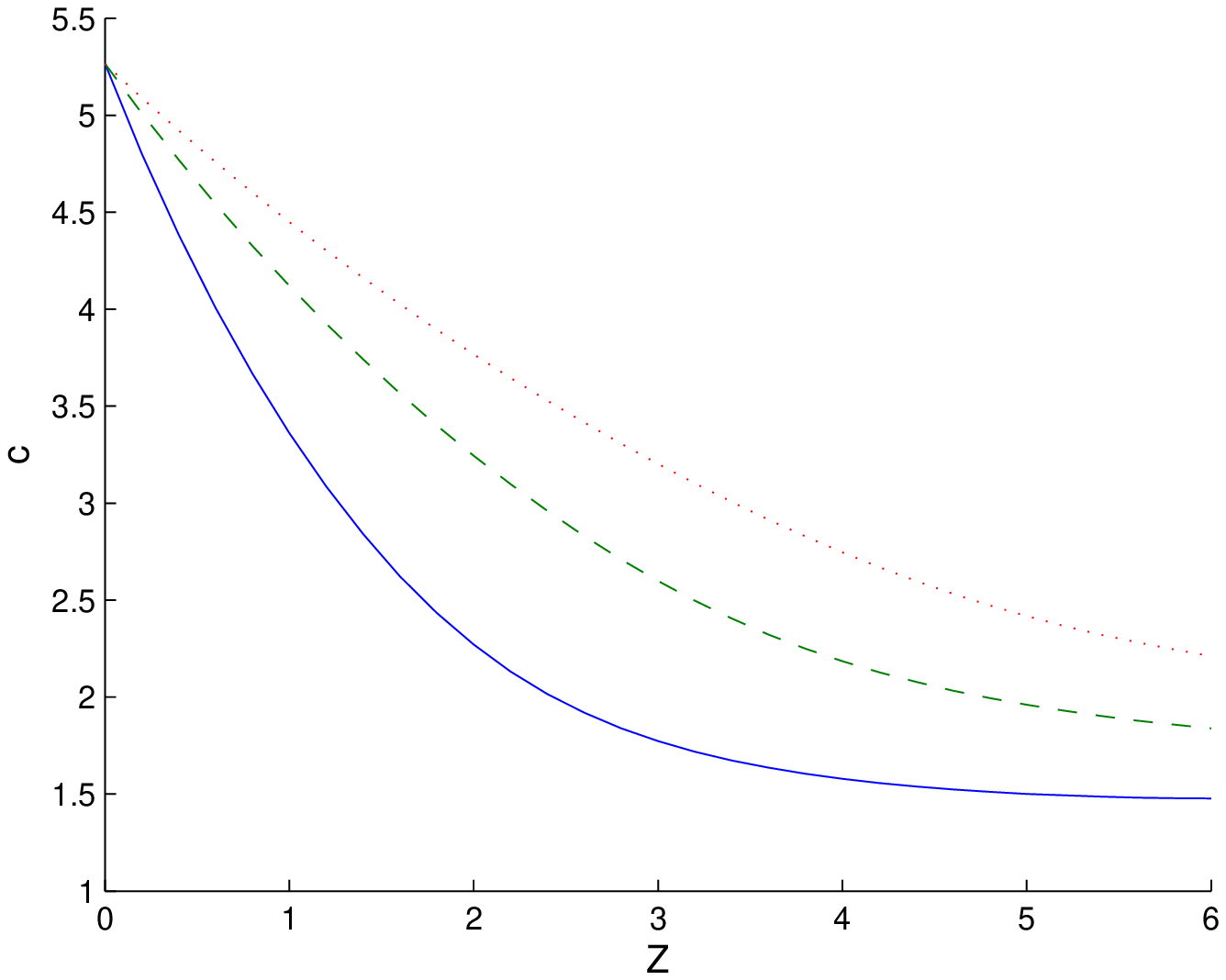}}  
\label{fig:subfigb}
}

\subfigure[]{
{\includegraphics[width=180pt]{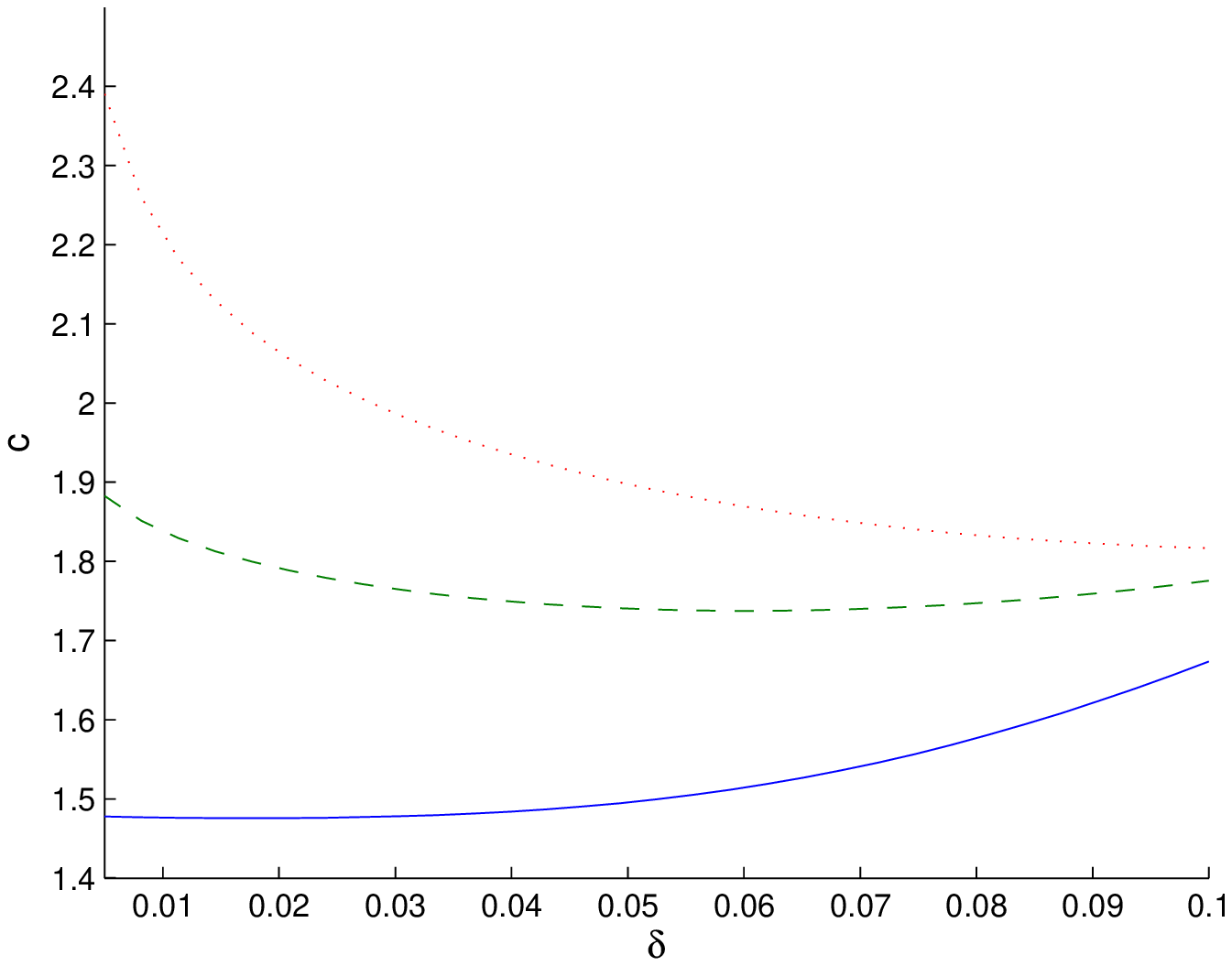}}  
\label{fig:subfigc}
}\subfigure[]{
{\includegraphics[width=180pt]{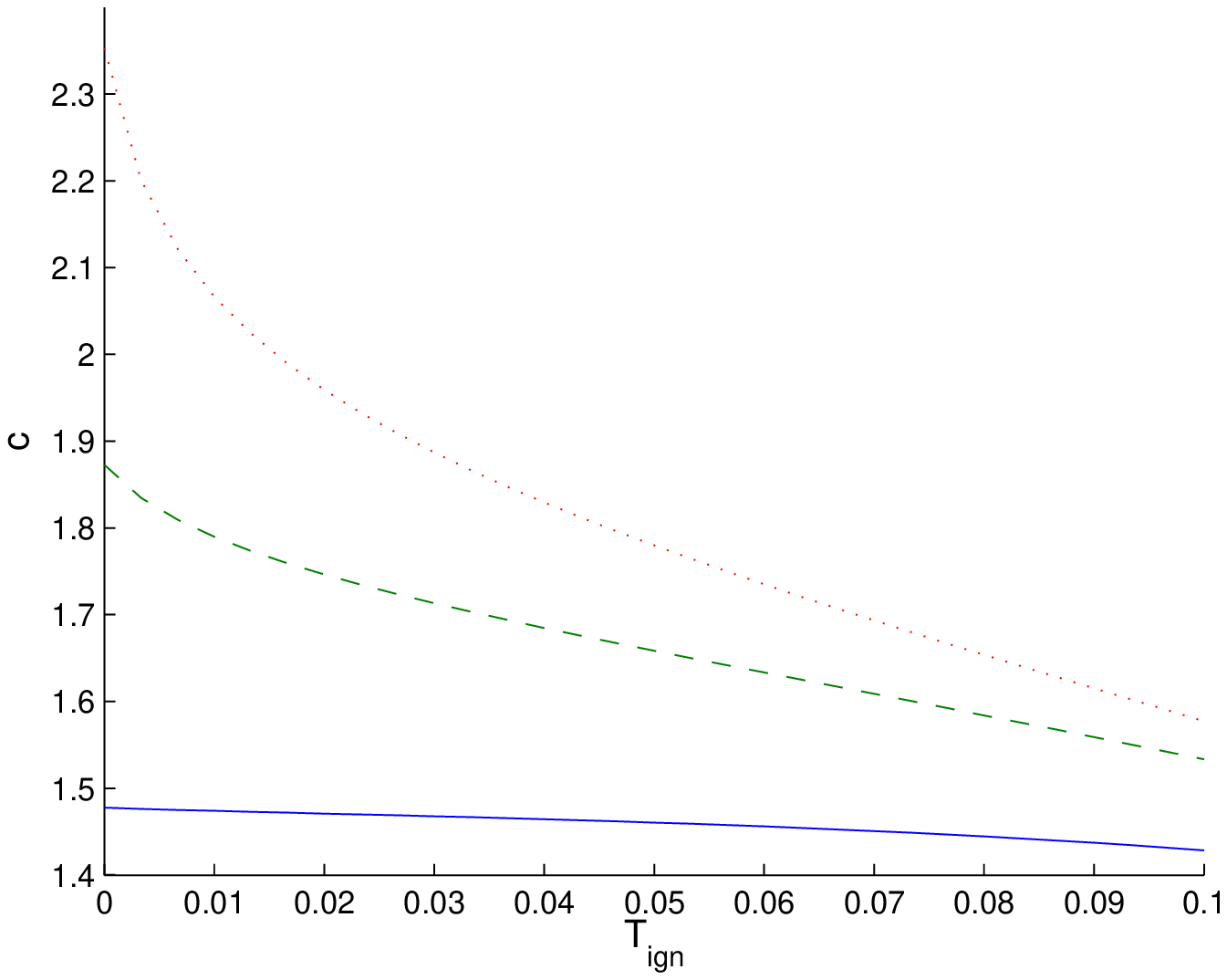}}  
\label{fig:subfigd}
}\caption{{{Plots showing the speed $c$ of the front as a function of the parameters (a) $h$ (with $Z=6$, $T_{ign}=0.003$, {and $\delta=0.01$}), (b) $Z$ (with $h=0.3$, $T_{ign}=0.003$, and $\delta=0.01$), (c) $\delta$ (with $h=0.3$, $Z=6$, and $T_{ign}=0.003$), and (d) $T_{ign}$ (with $h=0.3$, $Z=6$, {and $\delta=0.01$}). In each of the four graphs, the solid line corresponds to $\sigma=0.25$, the dashed line to $\sigma=0.5$, and the dotted line to $\sigma=0.75$. 
\label{SpeedFigure}
\vspace{-0.5cm}}}}
\end{figure}

\begin{figure}
\hspace{-0cm}
\scalebox{.73}{{\includegraphics{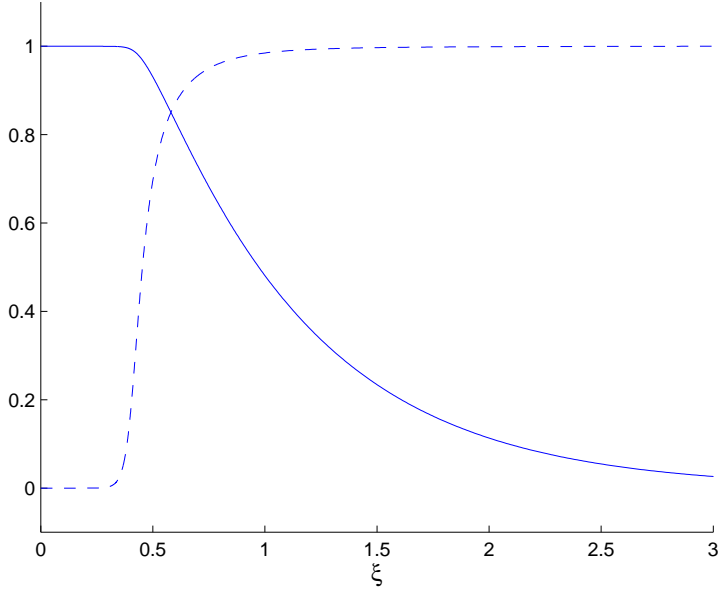}}}
\caption{\label{FigFront} Front solution of the system (\ref{front_eps0}) in the case $c=1.4684$, $h=0.3$, $Z=6$, $\sigma=0.25$, $T_{ign}=0.003$, and $\delta=0.01$. The solid line corresponds to {$u$ and the dashed line to $y$}.
}
\end{figure}

\subsection{Spectral Stability}
\label{SpSt}

We are now interested in numerically studying the point spectrum of the eigenvalue problem (\ref{EigProb}) on $L^2(\mathbb{R})$. The system (\ref{EigProb}) can be turned into a linear dynamical system 
of the form
\begin{equation}
X'=A(\xi,\lambda)\,X,
\label{linear}
\end{equation}
where $A$ is the $3\times 3$ following square matrix 
\eq{
A(\xi,\lambda)=
\left(
\begin{array}{ccccc}
0&&1&&0\\
\lambda-h\,F_w(\hw)\,\hy&&-c&& (1-h)F_w(\hw)\,\hy -F(\hw)\\
-h\,F_w(\hw)\,\hy/c&&0&& \lambda/c+(1-h)\,F_w(\hw)\,\hy/c +F(\hw)/c
\end{array}
\right).
}{A}
The asymptotic behavior as $\xi\rightarrow\infty$ of the solutions to (\ref{linear}) 
is determined by the matrix
$A^{\infty}(\lambda)=\lim_{\xi\rightarrow \infty}A(\xi,\lambda)$,
which is found by inserting the values $\hy=1$ and $\hw=0$ into (\ref{A})
\[
A^{\infty}(\lambda) = \begin{pmatrix}
0 & 1 & 0 \\
\lambda &  -c  &  0 \\
0 &  0  &  \lambda/c
\end{pmatrix}.
\]
%
For $\realpart{\lambda} > 0$, the unique eigenvalue of $A^{\infty}$ with negative real part and the corresponding eigenvector are  given by
$\mu_+ = \frac{-c - \sqrt{c^2 + 4 \lambda}}{2},\;\;v_+ = \left(1,\;\mu_+\;,0\right)^T.$
System (\ref{linear}) then has a unique solution $X_+$ satisfying \cite{codd}
$
\lim_{\xi\rightarrow \infty} X_{+}e^{-\mu_+\,\xi}=v_+.
$
In this situation where the linear system (\ref{linear}) has a one-dimensional stable manifold at $\infty$, the definition of the 
Evans function requires the adjoint system \cite{Pego92} 
\begin{equation}
Y'=-A^T(\xi,\lambda)\,Y.
\label{adjoint}
\end{equation}
The asymptotic behavior as $\xi\rightarrow-\infty$ of the solutions to (\ref{adjoint}) 
is determined by the matrix
$
A^{-\infty}(\lambda)=\lim_{\xi\rightarrow -\infty}A(\xi,\lambda)$,
which is found by inserting the values $\hy=0$ and $\hw=1$ into (\ref{A})
\[
A^{-\infty}(\lambda) =  \begin{pmatrix}
0 & 1 & 0 \\
\lambda &  -c  &  -e^{(1-h)Z} \\
0 &  0  &  (\lambda + e^{(1-h)Z})/c
\end{pmatrix}.
\]
For $\realpart{\lambda} > 0$, the unique eigenvalue of $-\left(A^{\infty}\right)^T$ with positive real part and the corresponding eigenvector are  given by
\eqnn{\mu_-=-\mu_+ = \frac{c + \sqrt{c^2 + 4 \lambda}}{2},\;\;
v_- = \left(-\lambda,\;\mu_-,\; \frac{c(c + \sqrt{c^2+4\lambda})}{c e^{Z(h-1)}(c + \sqrt{c^2+4\lambda} + 2 \lambda/c)+ 2 }\right)^T.}
System (\ref{adjoint}) then has a unique solution $Y_-$ satisfying
$\lim_{\xi\rightarrow -\infty} Y_{-}e^{-\mu_-\,\xi}=v_-.$
The Evans function $D(\lambda)$ can then be defined as \cite{Pego92}
$$
D(\lambda)=X_+(0)\cdot Y_-(0).
$$
The Evans function  is an analytic function in any region of the complex plane where $\mu_+$ is the eigenvalue of  $A^{\infty}$ that
has the smallest real part. In the region where $\mu_+$ is the unique eigenvalue with negative real part (which, in our case, is the open right half of the complex plane), the zeroes of the Evans function correspond to the point spectrum of (\ref{EigProb}) on $L^2(\mathbb{R})$. Note that the three eigenvalues of $A^{\infty}$ are  ${-c/2 \pm \sqrt{c^2 + 4 \lambda}}{/2}$ and $\lambda/c$, while the eigenvalues of $-(A^{\infty})^T$ are ${c/2 \pm \sqrt{c^2 + 4 \lambda}}{/2}$ and $-(\lambda+e^{Z(1-h)})/c$. This implies that to define a region where  the Evans function is analytic, it suffices to consider the values of $\lambda$ satisfying 
\begin{equation}
\realpart{\lambda}>-c^2/4.
\label{ancondition}
\end{equation}

To compute the Evans numerically, we choose a positive value $\xi=L$ at which the matrix given in (\ref{A}) is suitably close to its asymptotic value $A^{\infty}$. We then integrate the system (\ref{linear}) backward from $\xi=L$ in the direction of the eigenvector $v_+$ and find $X_+(0)$. In order to eliminate the exponential growth due to the eigenvalue with negative real part as we integrate from $\xi=L$, we modify the system in the following way  
$$
 X'=\left(A-\mu_+ I\right)\,X
 $$
 and use the initial condition $X(L)=v_+$. 
Similarly, we find $Y_-(0)$ by integrating the adjoint system (\ref{adjoint}).  Note that in the numerical computations, we choose the eigenvectors $v_\pm$ so that $v_+\cdot v_-=1$. This choice has the convenient consequence that $\lim_{\lambda\rightarrow \infty}D(\lambda)=1$ \cite{Pego92}.  

{In our case, since $F(w)=0$ for {$w< T_{ign}$}, the matrix $A$ in \eqref{A} reaches the constant matrix $A^{\infty}$ for a finite positive value of $\xi$. When computing $X_+$, we thus choose the value of $L$ to be the lowest value of $\xi$ satisfying $\hw(\xi)<T_{ign}$.}

Since we are interested in the zeroes of the Evans function, the standard method is to compute the integral of the logarithmic derivative of the Evans function on a given closed curve and obtain the winding number of $D(\lambda)$ along that curve. The eigenvalue problem (\ref{EigProb}) has an eigenvalue at $\lambda=0$ due to the translation invariance of the  system (\ref{system}). In order to numerically verify that there are no other zeroes of the Evans function on the right side of the complex plane, we use the bound provided by Theorem \ref{T2}. Once the quantity $h_3$ is computed numerically, one can choose a closed curve whose points satisfy (\ref{ancondition}) and whose interior encloses the  intersection of the   right side of the complex plane and the region defined by (\ref{rein126z}). For example, in the case  $h=0.3$, $Z=6$, $\sigma=0.25$, $T_{ign}=0.003$, and $\delta=0.01$ we find $h_3$ to be $144.6761$ and $c$ to be $1.4684$. Our numerical winding number computation then shows that the Evans function has no other zeroes than the one at the origin. 

We have made similar computations for other values of the parameters $h$, $Z$, $\sigma$, and $T_{ign}$. Every time, we have found the eigenvalue at the origin to be the only one. This thus strongly suggests that there is a regime in which the front solution is spectrally stable. 

As a simple example, Figure \ref{Fig3} shows the graph of the Evans function for values of $\lambda$ on the boundary of the right half of the circle of radius 3 and center -0.05  (with $h=0.3$, $Z=6$, $\sigma=0.25$, $T_{ign}=0.003$, and $\delta=0.01$). The numerical winding number is numerically computed to be 1.

\begin{figure}
\hspace{-0cm}
\scalebox{.73}{{\includegraphics{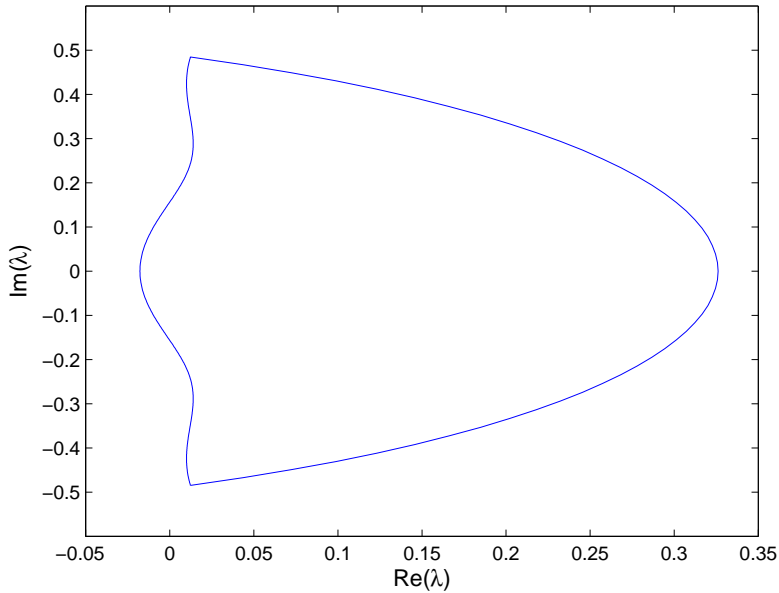}}}
\caption{\label{Fig3} Graph of the Evans function for values of $\lambda$ on the boundary of the right half of the circle of radius 3 and center -0.05  (with $h=0.3$, $Z=6$, $\sigma=0.25$, $T_{ign}=0.003$, and $\delta=0.01$). 
}
\end{figure}

\section{Interpretation of spectral information}
\label{nonl}

In this section we use the information about the spectrum obtained above to describe the nature of the instability of the front in the case when the front does not have unstable discrete spectrum, i.e.  the instability is caused by the continuous spectrum only.

Let $E_0$ denote $H^1$ or $BUC$, with norm  $\|\;\|_0$.  Let $E_\alpha$ be the corresponding weighted space with a weight  $ e^{\alpha\xi}$: $u \in E_\alpha$ if and only if $e^{\alpha\xi}u(\xi) \in E_0$, and $\|u\|_\alpha = \|e^{\alpha\xi}u(\xi)\|_0$.

An instability that can be removed by an exponential weight is called convective   \cite{Sandstede,ss00}. Indeed, if a perturbation grows while being convected  to, for example,  $- \infty$, then the  perturbation may stay bounded  or even decay in an exponentially weighted norm with a weight   that decays at  $-\infty$. 
 
A wave that is unstable in the space $E_0$ but stable in some weighted space $E_\alpha$,  in the literature is called 
 convectively   unstable.   {Convective instability of a wave  is characterized by the fact    that small perturbations to the wave  decay pointwise. }

Since the system \eqref{system} is partially parabolic, the operator  $\mathcal L$ of the linearization  of this system about the front is not sectorial and generates  not an analytic but a $C^0$-semigroup.
To interpret the absence of the unstable discrete spectrum for the full nonlinear system, one needs appropriate estimates on the semigroup. 
The semigroup estimates follow from the main theorem in \cite{GLS}. For  convenience  we make the theorem available  to the reader in Section \ref{semigroup}  of the Appendix.

It is important to  note that for the system \eqref{system}, the assumption of Theorem \ref{th:glslin} holds in the space $E_{\alpha}$ (see Lemma \ref{lem:essp}). 

In the weighted space the nonlinearity is not well defined, therefore  the semigroup estimates \eqref{semigroup} are  not sufficient to simply imply nonlinear stability in the weighted spaces.  Nevertheless, the system \eqref{system} and the front that it supports as a solution  have properties that allow to obtain nonlinear stability.  The following statement follows from Theorem 3.14 in \cite{gls10} which  is formulated in   Section \ref{ns} of the Appendix.

\begin{theorem}
\label{th:stab}
Consider the system \eqref{e:6} with  $c=c(h,Z, \gamma, \sigma)$ being the speed of the front $\widehat Y(\xi)=(\widehat u(\xi),\widehat y(\xi))$  and $h$, $Z$, $\gamma$, $\sigma$ be such that  the Evans function for the traveling wave  $\widehat Y(\xi)$  has no zeros in the half-plane $\mathrm{Re}\, \lambda \ge 0$ other than a simple zero at the origin.   Let $\alpha$  be chosen as in Lemma \ref{lem:essp}.
Suppose  the  initial conditions for \eqref{system}   be $(\widehat u(\xi)+\widetilde u^0(\xi), \widehat y(\xi)) + \widetilde y^0(\xi))$ with  $(\widetilde u^0(\xi)), \widetilde y^0(\xi)) \in E_{\alpha}^2$   and $\|(\widetilde u^0, \widetilde y^0)\|_{\alpha}$ small, and let $(u(t,\xi),y(t,\xi))$ be the solution of \eqref{system}  with $(u(0,\xi),y(0,\xi))= (\widehat u(\xi)+\widetilde u^0(\xi), \widehat y(\xi) + \widetilde y^0(\xi))$.  Then:
\begin{enumerate}
\item $(u(t,\xi), y(t,\xi))$ is defined for all $t\ge0$.
\item {$(u(t,\xi), y(t,\xi))=(\widehat u(\xi-q(t))+\widetilde u(t,\xi), \widehat y(\xi-q(t))+\widetilde y(t,\xi))$} with $(\widetilde u(t,\xi), \widetilde y(t,\xi))$ in a fixed subspace of $E_\eta^2$ complementary to the span of $(\widehat u^{\prime}, \widehat v^{\prime})$.  
\item $\|(\widetilde u(t), \widetilde y(t))\|_{0}+ |q(t)|$ is small for all $t\ge0$. 
\item $\|(\widetilde u(t), \widetilde y(t))\|_{\alpha}$ decays exponentially as $t\to\infty$.
\item There exists $q^*$ such that $|q(t)-q^*|$  decays exponentially as $t\to\infty$. 
\item There is a constant $C>0$ independent of $(\widetilde u^0, \widetilde y^0)$ such that $\|\widetilde u(t)\|_{0}\le C\|(\widetilde u^0,\widetilde y^0)\|_{\alpha}$ for all $t\ge0$.
\item  $\|\widetilde y( t)\|_{0}$ decays exponentially as $t\to\infty$.
\end{enumerate}
\end{theorem}

  The stability results of Theorem 3.14 from \cite{gls10} are  applicable to the  front in \eqref{e:6} because of   Lemma \ref{lem:essp}  and  the following properties of \eqref{e:6}
\begin{itemize}
\item[(i)] $R(u,y_-) =R(u,0)=\left .  \begin{pmatrix}yF(hu+(1-h)(1-y))\\-yF(hu+(1-h)(1-y))\end{pmatrix}\right|_{(u,0)}\equiv \begin{pmatrix}0\\0\end{pmatrix}$, where $R$ is the nonlinearity in \eqref{e:6} and $y_{-}=0$ is the $y$-component of the equilibrium to which the front converges at $-\infty$ (compare with Hypothesis~\ref{hypo1}).
\item[(ii)] The operator $\partial_{\xi\xi}+c\partial_\xi$ on $E_0$ generates a bounded semigroup.
\item[(iii)] The spectrum of the  operator $c\partial_{\xi} -\mathrm {e}^{(1-h)Z}$  on $E_0$ is  contained in $\mathrm{Re}\, \lambda \le - \mathrm {e}^{(1-h)Z}<0$ (compare with Hypothesis~\ref{hypo2}).
\end{itemize}
  

Under additional assumptions on the initial conditions, Theorem 3.16 in \cite{gls10}  (see Section \ref{ns}) implies some more  conclusions:   
\begin{theorem}
\label{th:stab2}
Consider the Banach space $E_0\cap L^1(\mathbb R)$ with the norm 
$\|u\|_{E_0\cap L^1(\mathbb R)}=\max\{\|u\|_{E_0},\|u\|_{L^1(\mathbb R)}\}.$  Suppose 
$(\widetilde u^0(\xi),\widetilde y^0(\xi))  \in (E_{\alpha}\cap L^1(\mathbb R))^2$  has norms  in $E_{\alpha}$ and $L^1$  which are sufficiently small, and as in Theorem \ref{th:stab}, let $(u(t,\xi),y(t,\xi))$ be the solution of \eqref{system}  with $(u(0,\xi),y(0,\xi))= (\widehat u(\xi)+\widetilde u^0(\xi), \widehat y(\xi) + \widetilde y^0(\xi))$.  Then  for all $t\ge0$:  
\begin{enumerate}
\item $(\widetilde u(t,\xi), \widetilde y(t,\xi))  \in (E_\alpha\cap L^1(\mathbb R))^2$.
\item  There is a constant $C$ independent of $(\widetilde u^0(\xi),\widetilde y^0(\xi))$ such that
\begin{eqnarray*}
&\|\widetilde y(t)\|_{L^1}\le& C\,\max\left\{\|\widetilde y^0\|_{L^1},\|(\widetilde u^0(\xi),\widetilde y^0(\xi))\|_\alpha \right\}, \\
&\|\widetilde y(t)\|_{L^\infty}\le& C\, \min\{1,t^{-\frac{1}{2}}\}\, \max\left\{\|\widetilde y^0\|_{L^1},\|(\widetilde u^0(\xi),\widetilde y^0(\xi))\|_\alpha \right\}.
\end{eqnarray*}
\item  $\|\widetilde u \|_{L^1}$ decays exponentially as $t\to\infty$.
\end{enumerate}
\end{theorem}

As we discussed in Section~\ref{reduction} system \eqref{system} could  also represent  \eqref{phys00_r}.  In that  case,  $u$ is  the scaled pressure $P$, $y$ is the concentration of the fuel $Y$, and  $T$ is a linear combination of those. Therefore Theorem \ref{th:stab} implies that small perturbations to fuel which are far from the unburned state decay exponentially fast, while perturbations to the pressure and temperature stay bounded or, under additional assumptions (Theorem \ref{th:stab2}),  decay at algebraic rates.  These results have clear physical meanings: if some fuel is added in the area where the temperature is high, the fuel will burn out; if  a cold or hot spot is introduced where there is no fuel left, the temperature will diffuse at algebraic rates as expected from the heat equation. 
In the original system  \eqref{system},  $u$ represents a linear combination of the pressure and temperature, and $y$ represents the concentration of the  fuel, so a similar interpretation makes physical sense.

\section{Conclusions and extension of the results}
\label{Conc}

In this article, we have considered front solutions to the model \eqref{phys00}. The fronts were seen as solutions of a reduced system 
\eqref{system} obtained by from \eqref{phys00} by implementing the scaling \eqref{ls0} and requiring the initial conditions \eqref{incond2}. We have used a combination of energy estimates computations and numerics based on the Evans functions to show that there is a parameter regime for which there are no unstable eigenvalues. We then used recently obtained results about partially parabolic systems and about fronts with marginal spectrum to prove the nonlinear stability of the fronts in the absence of unstable spectrum. Our results are concerned with the case where the nonlinearity $F$ takes the form \eqref{Fdefrc} with $\delta$ very small, i.e. very close to the discontinuous version \eqref{Fdefr}. 

We now show that our result extends to the general case where the condition \eqref{incond2} does not necessarily hold, i.e. it applies to the fronts as solutions of \eqref{phys00}. 
In general, in the case $\epsilon=0$, with the scaling \eqref{ls0}, the system  \eqref{phys00} becomes
\eq{
m_t&=0,\\
u_t&= u_{xx}+yF(m+hu+(1-h)(m-y)),  \\
y_t&= \bar \eps y_{xx} -yF(m+hu+(1-h) (m-y)), 
}{rs}
where $m\equiv y+v$. Under the initial condition \eqref{incond2}, we have that $m=1$  and the system \eqref{rs} becomes the reduced system \eqref{system}. To show
that our stability results concerning  the reduced system carry over to the general case \eqref{rs}, 
 we first write the eigenvalue problem arising from the linearization of \eqref{rs} about the front solutions
\begin{equation}
\label{EigProbr}
\begin{aligned}
\lambda r &=c r_\xi,\\
\lambda p &= p_{\xi\xi} + c p_\xi + F_w(\hw)\,\hy\,(hp-(1-h)(q-r))+F(\hw)q, \\
\lambda q &= c q_\xi-F_w(\hw)\,\hy\, (hp-(1-h)q)-F(\hw)(q-r),
\end{aligned}
\end{equation}
where we have used the fact that, in view of the asymptotic values for the front \eqref{bc}, we have  $\widehat{m}\equiv \hy+\hv=1$.
The equation above  generalizes \eqref{EigProb} obtained for the reduced system. Now, if $\lambda$ is in the point spectrum, the first equation 
of  \eqref{EigProbr} implies that $r=0$, thus reducing the problem to the one we studied for the reduced system. Therefore, the point spectra of \eqref{EigProb} and \eqref{EigProbr} are the same.
As for the essential spectrum, if we apply the procedure described in Section \ref{Ess} to  \eqref{EigProbr}, it is easy to see that the spectrum is unchanged with the exception that the presence of the equation for $r$ adds one more covering of the imaginary axis. 
 
The nonlinear stability results similar to ones  described in Theorems~\ref{th:stab} and \ref{th:stab2} can be obtained from   Theorems~\ref{th:nl1} and  \ref{th:nl2},  since in the diagonal  operator 
$$L_1=\begin{pmatrix}c\partial_{\xi}&0 \\ 0&D_1 \partial_{\xi\xi}+c\partial_\xi  \end{pmatrix}$$
 generates a bounded semigroup.  Indeed, the operator $c\partial_{\xi}$ generates the semigroup $P(t)s(\xi) = s(\xi+ct)$, which is  a $C^0$-semigroup. The operator $\partial_{\xi\xi} +c\partial_{\xi} $  on $E_0$ is sectorial (see \cite{Henry81}, pp. 136-137, and \cite{Pazy}, Section 3.2, Corollary 2.3) and hence generates an analytic semigroup (\cite{Henry81}, Theorem 1.3.4). 
 
 The behavior of the perturbation $\widetilde m$ to $\widehat m$ mimics the behavior of the perturbation $\widetilde u$ to the $\widehat u$ component of the front: $\|\widetilde m \|_0$ is bounded by a constant defined by the norm of the initial perturbations to the front and $\|\widetilde m \|_{\alpha}$ converges exponentially to 0 as $t\to \infty$. Taking into account exponential convergence of the $\widetilde y$-component to $0$ in both norms, we can conclude that $\widetilde v=\widetilde m-\widetilde y$ which is  the perturbation to the pressure component in the model stays bounded (and controlled by the size of the initial conditions)  in the $\|\cdot  \|_0$-norm and converges  to $0$ in $\|\cdot \|_{\alpha}$-norm.



\section{Appendix}

\subsection{Semigroup Estimates}\label{semigroup}

For our system that estimate follows from the following theorem \cite{GLS}.
Let the linearization of a PDE system about a wave (a front or a pulse) traveling with velocity $c$  have the form
$${
\begin{split}
\partial_tU&=D\partial_{\xi\xi}U+(\widetilde A +diag(c,\ldots,c))\partial_\xi U+B_{11}(\xi)U+B_{12}(\xi)V,  \\
\partial_tV&=c\partial_\xi V+B_{21}(\xi)U+B_{22}(\xi)V, 
\end{split}
}$$
where $U$ is  a $k\times 1$  vector,  $V$  a $(n-k)\times 1$ vector,   $D$ and  $\widetilde A$  are constant matrices,  
$D=\mathrm{diag}(d_1,\ldots,d_k)$, 
all $d_i>0$, $1 \le k\leq n$,  and 
 each $B_{ij}(\xi)$ exponentially approaches a constant matrix  $B_{ij}^\pm$ as $\xi\to\pm\infty$. Consider the linear operator $\mathcal L$ associated with the differential expression
$$
 \mathcal L=\begin{pmatrix}
 D\partial_{\xi\xi}+(\widetilde A +diag(c,\ldots,c))\partial_\xi+B_{11}&B_{12}\\B_{21}&c\partial_\xi+B_{22}
 \end{pmatrix}
$$
on one of the Banach spaces $L^2(\mathbb R)^n$, $H^1(\mathbb R)^n$, $L^1(\mathbb R)^n$, or $BUC(\mathbb R)^n$.
  
\begin{theorem} \label{th:glslin} \cite{GLS}
Suppose the spectrum of $\mathcal L$ is contained in $\mathrm {Re}\,\lambda\le-\nu$, $\nu>0$, except for an eigenvalue at $0$ of finite algebraic multiplicity.  Let $\mathcal{P}_0$ be the Riesz spectral projection for  $\mathcal L$  whose kernel is equal to  the generalized eigenspace for the $0$ eigenvalue, and let $0<\delta<\nu$. 
 Then there is a number $K>0$ such that  
$$\|e^{t\,\mathcal L\, \mathcal P_0} \|  \le Ke^{-\delta t}.$$
\end{theorem}

The proof of Theorem \ref{th:glslin} is based on the Gearhart-Pr\"uss  or Greiner Spectral Mapping Theorem, for dealing with Hilbert space and Banach space, respectively.

\subsection{Nonlinear Stability}\label{ns}

The following theorem  is formulated and proved in \cite{gls10}.
Consider the system 
\begin{equation}
\label{geneq}
Y_t =DY_{xx}+ R(Y), 
\end{equation}
with $x \in \mathbb R$, $t \ge 0$, and $Y: \mathbb R\times\mathbb R_+\to \mathbb R^n$. The function $R: \mathbb R^n\to\mathbb R^n$ is smooth, $D=\mbox{diag}(d_i)$ is an $n \times n$ constant diagonal matrix with $d_i > 0$ for $i=1,\ldots,k$, where $k$ a number between 1 and $n$, and $d_i = 0$ otherwise.  
Let  
\begin{equation}
\label{movingeq}
Y_t =DY_{\xi\xi}+cY_\xi+ R(Y), 
\end{equation}
be the system that has  as a solution a front or a pulse  $Y^*$  that moves  with the velocity $c$ and  that asymptotically connects the equilibrium $Y_-$ at $-\infty$ to the equilibrium $Y_+$ at $+\infty$.  Without loss of generality we take $Y_-$ to be a zero vector.  

 Assume that,  in \eqref{movingeq},  $Y$ can be written as $Y=(U,V) \in \mathbb R^{n}$,   where $U\in \mathbb R^k$,  $V\in \mathbb R^{n-k}$,  and the splitting is such that    there exists a $ k \times k$ matrix $A_1$ such that $R(U, 0) = (A_1 U, 0)$, 
 and the equation \eqref{geneq} has the form 
\begin{align}
U_t &= D_1 U_{\xi\xi}+cU_{\xi}+R_1(U,V), 
\notag
 \\
V_t &= D_2 V_{\xi\xi}+cV_{x}+R_2(U,V), 
\notag
\end{align}
where $D_1$ and $D_2$ nonnegative diagonal matrices, and $R_1$ and $R_2$ are continuously differentiable maps.
 The linearization of  \eqref{movingeq}  then  at $Y_-=(0,0)$, reads
\begin{align*}
U_t &=L^{(1)}U +D_VR_1(0,0)V:= D_1 U_{\xi\xi}+cU_\xi+A_1U +D_VR_1(0,0)V, \\
V_t &=L^{(2)}V:= D_2 V_{\xi\xi}+cV_\xi+D_VR_2(0,0)V,
\end{align*}
Additional  assumptions can be formulated as: 
\begin{hypothesis}\label{hypo1}
Assume that the traveling wave  $Y_*$  is spectrally stable in an exponentially weighted space  $E_\alpha$.  
\end{hypothesis}

\begin{hypothesis}\label{hypo2}
Assume that  
(i)   the operator $L^{(1)}$ on $E_0$ generates a bounded
semigroup, and 
(ii) the operator  associated with $L^{(2)}$ on $E_0$ has its spectrum in the half-plane  $\mathrm{Re}\,\lambda \le -\nu$ for some $\nu>0$.   
\end{hypothesis}
The following theorem then  holds.

\begin{theorem} \label{th:nl1}
Assume Hypotheses \ref{hypo1} and \ref{hypo2}.  Then the wave $Y_*$ is nonlinearly convectively unstable, with $\alpha$ given by Hypothesis \ref{hypo1}.  Thus if the perturbation   $(\widetilde U(\cdot,0),\widetilde V(\cdot,0))$ of the traveling wave is initially small in $E_0 \cap E_\alpha$, then the corresponding solution of \eqref{movingeq} decays exponentially in $E_\alpha$ as $t\to\infty$ to a particular shift of the wave.  The solution can be written  as
$$
(U,V)(\xi,t)=(U_*(\xi+ q(t))+\widetilde U(\xi,t),V_*(\xi+q(t))+\widetilde V(\xi,t))
$$
where, for each $t$,  the function $(\widetilde U(\cdot,t),\widetilde V(\cdot,t))$ belongs to a fixed subspace of $E_0 \cap E_\alpha$ complementary to the the span of $Y_*^\prime(\cdot)$. $\widetilde U(\cdot,t)$ stays small in $E_0$, while $\widetilde V(\cdot,t)$ decays exponentially in $E_0$ as $t\to\infty$, and  the function $q(t)$ converges exponentially to a finite limit as $t\to \infty$. 
\end{theorem}

\begin{theorem} \label{th:nl2}
In addition to the assumptions in Theorem \ref{th:nl1}, let us suppose that the linear equation $U_t=L^{(1)}U$ is parabolic, i.e., the diagonal entries of $D_1$ are all positive.   If the initial perturbation of the traveling wave is also small in $L^1$-norm, then $\widetilde U(\cdot,t)$ stays small in the $L^1$-norm and decays like $t^{-\frac{1}{2}}$ in the $L^\infty$-norm
as $t\to\infty$.
\end{theorem}


\end{document}